\documentclass[10pt, conference]{IEEEtran}
\usepackage{subfigure}
\usepackage{setspace}
\usepackage{amsmath}
\usepackage{amssymb}
\usepackage{amsfonts}
\usepackage{amscd}
\usepackage{mathrsfs}
\usepackage[final]{graphicx}
\usepackage{graphicx}
\usepackage{psfrag}
\usepackage{color}
\usepackage{url}

\newtheorem{theorem}{Theorem}

\newtheorem{definition}{Definition}

\newtheorem{example}{Example}

\title{Training-Embedded, Single-Symbol ML-Decodable, Distributed STBCs for Relay Networks}
\begin{document}
\author{
\authorblockN{J. Harshan}
\authorblockA{Dept. of ECE,\\ Indian Institute of Science \\
Bangalore 560012, India\\
{\em Email:} \url{harshan@ece.iisc.ernet.in}\\
}
\and
\authorblockN{B. Sundar Rajan}
\authorblockA{Dept. of ECE, \\Indian Institute of Science \\
Bangalore 560012, India\\
\hspace{0.2cm} {\em Email:} \url{bsrajan@ece.iisc.ernet.in}\\
}
\and
\authorblockN{Are Hj{\o}rungnes}
\authorblockA{UNIK-University Graduate Center \\
University of Oslo, \\NO-2027, Kjeller, Norway\\
{\em Email:} \url{arehj@unik.no}\\
}
}

\maketitle

%
\begin{abstract}

Recently, a special class of complex designs called Training-Embedded Complex Orthogonal Designs (TE-CODs) has been introduced to construct single-symbol Maximum Likelihood (ML) decodable (SSD) distributed space-time block codes (DSTBCs) for two-hop wireless relay networks using the amplify and forward protocol. However, to implement DSTBCs from square TE-CODs, the overhead due to the transmission of training symbols becomes prohibitively large as the number of relays increase. In this paper, we propose TE-Coordinate Interleaved Orthogonal Designs (TE-CIODs) to construct SSD DSTBCs. Exploiting the block diagonal structure of TE-CIODs, we show that, the overhead due to the transmission of training symbols to implement DSTBCs from TE-CIODs is smaller than that for TE-CODs. We also show that DSTBCs from TE-CIODs offer higher rate than those from TE-CODs for identical number of relays while maintaining the SSD and full-diversity properties. 
While TE-CODs offer full-diversity for arbitrary complex constellations, TE-CIODs offer full-diversity for any constellation when appropriately rotated.
\end{abstract}
\begin{keywords}
Cooperative diversity, single-symbol ML decoding, distributed space-time coding, complex orthogonal designs.
\end{keywords}

\section{Introduction and Preliminaries}
\label{sec1}
\begin{LARGE}D\end{LARGE}istributed space-time coding has been a powerful technique for achieving spatial diversity in wireless networks with single antenna terminals in two hop wireless networks \cite{LaW,JiH1}. The technique involves a two phase protocol where, in the first phase, the source broadcasts the information to the relays and in the second phase, the relays linearly process the signals received from the source and forward them to the destination such that the signal at the destination appears as a Space-Time Block Code (STBC). Such STBCs, generated distributively by the relay nodes, are called Distributed Space-Time Block Codes (DSTBCs).\\
\indent In a co-located Multiple-Input Multiple-Output (MIMO) system, an STBC is said to be Single-Symbol Maximum Likelihood (ML) Decodable (SSD) if the ML decoding metric splits as a sum of several terms, with each term being a function of only one of the information symbols \cite{KhR}. A DSTBC is said to be SSD if the STBC seen by the destination from the set of relays is SSD. Since the work of \cite{LaW, JiH1}, significant efforts have been made to design SSD DSTBCs. Towards that direction, SSD DSTBCs have been proposed for cooperative networks in \cite{YiK}, \cite{HaR}, \cite{YiK2} and \cite{SCR}. 



\indent It is well known that DSTBCs from complex orthogonal designs (CODs) \cite{TJC, Lia}, (both square and non-square CODs other than the Alamouti design), coordinate interleaved orthogonal designs (CIODs) \cite{KhR} and Clifford unitary weight designs (CUWDs) \cite{KaR} lose their single-symbol ML decodable (SSD) property when used in two-hop wireless relay networks using the amplify and forward protocol. In \cite{HRH}, a new class of high rate, training-embedded (TE) SSD DSTBCs are proposed from TE-CODs. The constructed codes include the training symbols in the structure of the code which has been shown to be the key point to obtain high rate along with the SSD property. The authors of \cite{HRH} show that non-square TE-CODs provide higher rates (in symbols per channel use) compared to the known SSD DSTBCs \cite{YiK2}, \cite{SCR} for relay networks when the number of relays is less than $10.$ Note that, the known codes in \cite{YiK2}, \cite{SCR} and \cite{HRH} (non-square TE-CODs) have exponential decoding delay and hence, in this paper, we focus on constructing SSD DSTBCs with low delay only.\\
\indent In order to implement a square TE-COD in a network with $2^{a}$ relays, a total of 
\begin{equation*}
N^{t}_{\mbox{TE-COD}} = \lceil \frac{2^{a} - a - 1}{2} \rceil + 2^{a} - a - 1
\end{equation*}
channel uses are required for transmitting the training symbols \cite{HRH} ($\lceil \frac{2^{a} - a - 1}{2} \rceil$ and $2^{a} - a - 1$ channel uses are required in the first phase and second phase respectively). However, the number of information symbols embedded in such a design is $a+1$ and hence, the number of channel uses required for transmitting the training symbols per information symbol is given by
\begin{equation}
\label{tr_ovh_te_cod}
\frac{\lceil \frac{2^{a} - a - 1}{2} \rceil + 2^{a} - a - 1}{a+1}.
\end{equation}
From \eqref{tr_ovh_te_cod}, it is clear that the overhead due to the transmission of training symbols (overhead both in-terms of power and bandwidth) increases as the number of relays increase, which is one of the drawbacks of implementing square TE-CODs as TE DSTBCs. Also, the number of complex symbols that a TE-COD for $2^{a}$ relays can accommodate is only $a + 1$ (which is same as that of a COD for $2^{a}$ antennas), which is a drawback of square TE-CODs. Therefore, the rate of TE-DSTBCs from TE-CODs (in symbols per channel use) when employed as in \cite{HRH} is given by
\begin{equation}
\label{rate_te_cod}
R_{\mbox{TE-CODs}} = \frac{a + 1}{a + 1 + \lceil \frac{2^{a} - a - 1}{2} \rceil + 2^{a}}
\end{equation}
wherein $a + 1 + \lceil \frac{2^{a} - a - 1}{2} \rceil$ and $2^{a}$ channel uses are used in first phase and second phase respectively. Note that $R_{\mbox{TE-CODs}}$ decreases exponentially with the number of relays, $2^{a}$.

\indent In this paper, we propose training embedded SSD DSTBCs for relay networks with rates higher than that of DSTBCs from TE-CODs (given in \eqref{rate_te_cod}). In particular, we employ linear precoding of information symbols at the source \cite{HaR} and use CIODs of \cite{KhR} instead of CODs to obtain a class of high-rate SSD DSTBCs. The main contributions of this paper can be summarized as follows:
\begin{itemize}
\item We employ precoding of information symbols at the source \cite{HaR} to construct high rate, low-delay, SSD DSTBCs for two-hop wireless relay networks based on the amplify and forward protocol. On the similar lines of \cite{HRH}, the proposed method has an in-built training scheme for the relays to learn the phase components of their backward channels which is shown to be the key point to obtain the SSD property. 
\item When all the zero entries of a COD (square or non-square) is replaced by a constant, the resulting design is called a Training-Embedded-CODs (TE-CODs) \cite{HRH}. Using square TE-CODs as ingredients, we construct TE-CIODs using the coordinate interleaved variables. Unlike TE-CODs, not all the entries of a TE-CIOD are non-zero. In particular, TE-CIODs have a block diagonal structure. 
\item It is well known that the number of complex variables that a CIOD can accommodate ($2a$ variables for $2^{a}$ antennas) is more than that of a COD ($a+1$ variables for $2^{a}$ antennas) for the same number of antennas \cite{KhR}. As a result, TE-CIODs continue to have larger number of information variables than TE-CODs. Exploiting the block diagonal structure of TE-CIODs, we show that the minimum number of training symbols required to implement a TE-CIOD as a SSD TE-DSTBC in a wireless network with $2^{a}$ relays is, 
\begin{equation*}
\lceil \frac{2^{a-1} - a}{2} \rceil + 2^{a-1} - a	
\end{equation*}
which is lesser than the number required for implementing TE-CODs for the same number of relays (which is given by $N^{t}_{\mbox{TE-COD}}$). Considering (i) the number of channel uses for transmitting the training symbols and (ii) the number of complex symbols in the design, TE-CIOD, we show that the rate of TE-DSTBCs from TE-CIODs is 
\begin{equation}
\label{rate_te_ciod}
R_{\mbox{TE-CIODs}} = \frac{2a}{2a + \lceil \frac{2^{a-1} - a}{2} \rceil + 2^{a}}.
\end{equation}
Hence, comparing \eqref{rate_te_ciod} with \eqref{rate_te_cod}, TE-DSTBCs from TE-CIODs provide higher rates (in symbols per channel use) compared to TE-DSTBCs from TE-CODs for a specified number of relays in a two-hop network, while retaining the SSD and full-diversity property. We highlight that the above rate advantage comes mainly from the block diagonal structure of TE-CIODs.
\end{itemize}
\indent \textit{Notations:} Throughout the paper, boldface letters and capital boldface letters are used to represent vectors and matrices respectively. For a complex matrix $\textbf{X}$, the matrices $\textbf{X}^*$, $\textbf{X}^T$,  $\textbf{X}^{H}$, $|\textbf{X}|$, $\mbox{Re}~\textbf{X}$ and $\mbox{Im}~\textbf{X}$ denote, respectively, the conjugate, transpose, conjugate transpose, determinant, real part and imaginary part of $\textbf{X}$. The element in the $r_1$-th row and the $r_2$-th column of the matrix $\textbf{X}$ is denoted by $[\textbf{X}]_{r_1,r_2}$. The $ T\times T$ identity matrix and the $T \times T$ zero matrix are respectively denoted by $\textbf{I}_T$ and $\textbf{0}_{T \times T}$. The magnitude of a complex number $x$, is denoted by $|x|$ and $E \left[x\right]$ is used to denote the expectation of the random variable $x.$ A circularly symmetric complex Gaussian random vector, $\textbf{x},$ with mean $\boldsymbol{\mu}$ and covariance matrix $\mathbf{\Gamma}$ is denoted by $\textbf{x} \sim \mathcal{CSCG} \left(\boldsymbol{\mu}, \mathbf{\Gamma} \right) $. The set of all integers, the real numbers and the complex numbers are respectively, denoted by ${\mathbb Z}$, $\mathbb{R}$ and ${\mathbb C}$ and  $\bf{i}$ is used to represent $\sqrt{-1}.$

The remaining content of the paper is organized as follows: The system model for our training-embedded precoded distributed space-time coding is described in Section \ref{sec2} which differs from the model of \cite{HRH} due to the precoding at the source. Construction of  TE-CIODs is presented in Section \ref{sec3}. The  SSD property, full-diversity property and comparison of rates with TE-CODs are discussed in Section \ref{sec4}. Simulation results are presented in Section \ref{sec5} and Section \ref{sec6} constitutes a short summary and possible directions for further research.
\begin{figure}
\centering
\includegraphics[width=3in]{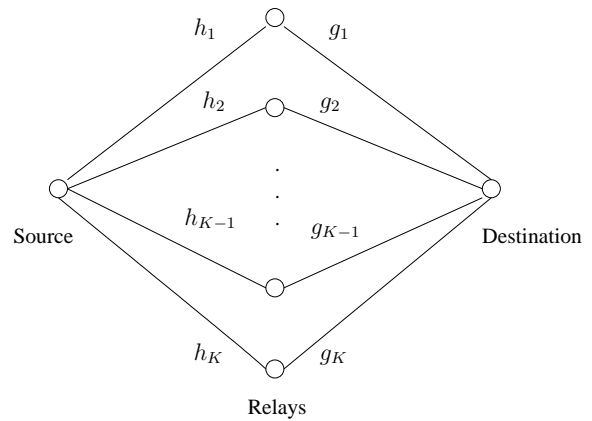}
\caption{Two-hop wireless relay network model.}
\label{model_network}
\end{figure}
\section{System Model}
\label{sec2}
The wireless network considered, as shown in Fig. \ref{model_network}, consists of $K + 2$ nodes, each having a single antenna. There is one source node and one destination node. All the other $K$ nodes are relays. We denote the channel from the source node to the $\lambda$-th relay as $h_{\lambda}$ and the channel from the $\lambda$-th relay to the destination node as $g_{\lambda}$ for $\lambda=1,2, \cdots, K$.
\noindent The following assumptions are made in our model:

\begin{itemize}
\item All the nodes are half duplex constrained.
\item Fading coefficients  $h_{\lambda}$ and $g_{\lambda}$ are i.i.d. $ \mathcal{CSCG} \left(0,1 \right)$ with a coherence time interval of at least $N$ and $T$ channel uses respectively, where $N$ is the number of channel uses for the transmission from the source to the relays and $T$ is the number of channel uses for transmission from the relays to the destination.
\item All the nodes are synchronized at the symbol level.
\item Relay nodes have the knowledge of only the phase components of the fade coefficients $h_{\lambda}$.
\item The destination knows all the fading coefficients $g_{\lambda}$, $h_{\lambda}$ for $\lambda = 1, 2, \cdots K,$ exactly.
\end{itemize}

\begin{figure*} 
\begin{equation} 
\label{concatenation}
\bar{\textbf{r}}_{\lambda} =
e^{-\textbf{i}\angle{h_k}}  \left[e^{\textbf{i}2 (\angle{\alpha}+\angle{h_k})}\textbf{r}_{\lambda}^{*}(1) ~e^{\textbf{i}2 (\angle{\alpha}+\angle{h_k})}\textbf{r}_{\lambda}^{*}(2) ~\cdots ~e^{\textbf{i}2 (\angle{\alpha}+\angle{h_k})}\textbf{r}^{*}_{\lambda}
(2^{a-1}-a) ~\textbf{r}^T_{\lambda} \right]^{T} \in \mathbb{C}^{(2^{a-1} + a) \times 1}
\end{equation}.
\end{figure*}
\begin{figure*}
\begin{equation}
\label{code_word}
\textbf{X}  = \left[
\textbf{A}_{1}\bar{\textbf{x}} + \textbf{B}_{1}\bar{\textbf{x}}^{*}  ~~
\textbf{A}_{2}\bar{\textbf{x}} + \textbf{B}_{2}\bar{\textbf{x}}^{*}  ~~
\cdots ~~
\textbf{A}_{\lambda}\bar{\textbf{x}} + \textbf{B}_{\lambda}\bar{\textbf{x}}^{*}
\right] \in \mathbb{C}^{2^a \times 2^a}.
\end{equation}
\end{figure*}
\begin{figure*}
\begin{equation}
\label{covariance_matrix}
\textbf{R} = {\frac{P_{2}2^a}{P_{r}}}\left[ \sum_{k=1}^{2^a} |g_{\lambda}|^{2}\left\lbrace \textbf{A}_{\lambda}\textbf{A}_{\lambda}^{H} + \textbf{B}_{\lambda}\textbf{B}_{\lambda}^{H}\right\rbrace \right]  + \textbf{I}_{T} \in \mathbb{C}^{2^a \times 2^a}.\\
\end{equation}
\end{figure*}
\begin{figure*}
\begin{equation}
\label{ML}
\hat{\textbf{x}} = arg\, \min_{\textbf{x} \in \mathcal{S}} \left[ -2 \mbox{Re} \left(\sqrt{\frac{P_{1}P_{2}N2^a}{P_{r}}}\textbf{g}^{H}\textbf{X}^{H}\textbf{R}^{-1}\textbf{y}\right) + {\frac{P_{1}P_{2}N2^a}{P_{r}}}\textbf{g}^{H}\textbf{X}^{H}\textbf{R}^{-1}\textbf{X}\textbf{g} \right] \in \mathbb{C}^{2a \times 1}.
\end{equation}
\hrule
\end{figure*}
\noindent The source is equipped with a codebook $\mathcal{S}$ = $\left\{ \textbf{x}_{1},\, \textbf{x}_{2},\, \textbf{x}_{3},\, \cdots, \textbf{x}_{L} \right\} $ consisting of information vectors $\textbf{x}_{l} \in \mathbb{C}^{N \times 1}$ such that $E\left[\textbf{x}_{l}^{H}\textbf{x}_{l}\right]$ = $1$. For this scenario, in \cite{HRH} we proposed TE-CODs for arbitrary values of $K.$ In this paper, we show that for the values $K=2^a,$ for any positive integer $a,$ we can use CIODs of \cite{KhR} to obtain TE-CIODs which have higher rate than TE-CODs for identical number of relays.

In systems employing CIODs with $2^a$ relays, the number of information symbols transmitted is $2a,$ and the information vectors are of the form,
\begin{equation*}
\textbf{x} = [ \underbrace{\alpha ~\alpha ~\cdots ~\alpha}_{\lceil \frac{2^{a-1}-a}{2} \rceil ~\mbox{times}} ~x_{1} ~x_{2} \cdots ~x_{2a}]^T \in \mathbb{C}^{N \times 1}, 
\end{equation*}
where the complex variables $x_{1}, x_{2} \cdots x_{2a}$ take values from a complex signal set denoted by $\mathcal{M},$  $\alpha \in \mathbb{C}$ is a non-zero complex constant chosen as the training symbol and $N= \lceil \frac{2^{a-1}-a}{2} \rceil+2a.$  The value of $\alpha$ is chosen such that the condition $E\left[\textbf{x}_{l}^{H}\textbf{x}_{l}\right]$ = 1 is satisfied. The value of $\alpha$ is assumed to be known to all the relays and the destination.

The source is also equipped with a pair of $N \times N$ matrices $\textbf{P}$ and $\textbf{Q}$ called precoding matrices. Every transmission from the source to the destination comprises of two phases. When the source needs to transmit an information vector $\textbf{x} \in \mathcal{S}$ to the destination, it generates a new vector $\hat{\textbf{x}}$ as,
\begin{equation}
\label{cord_int}
\hat{\textbf{x}} =  \textbf{P}\textbf{x} + \textbf{Q}\textbf{x}^{*}=[ \underbrace{\alpha ~\alpha ~\cdots ~\alpha}_{\lceil \frac{2^{a-1}-a}{2} \rceil ~\mbox{times}} ~\hat{x}_{1} ~\hat{x}_{2} \cdots ~\hat{x}_{2a}]^T \in \mathbb{C}^{N \times 1}
\end{equation}
\noindent
where the precoding matrices satisfy the condition $E\left[\hat{\textbf{x}}^{H}\hat{\textbf{x}}\right]$ = $1$ and broadcasts the vector $\hat{\textbf{x}}$ to all the $2^a$ relays (but not to the destination which is assumed to be located far from the source). The precoding matrices are chosen such that the linear processing is performed only on the information symbols but not on the training symbols. The received vector at the $\lambda$-th relay is given by $\textbf{r}_{\lambda} = \sqrt{P_{1}N}h_{\lambda}\hat{\textbf{x}} + \textbf{n}_{\lambda} \in \mathbb{C}^{N \times 1}$, for all $\lambda = 1,2,\cdots, 2^a$ where $\textbf{n}_{\lambda} \sim \mathcal{CSCG} \left(\textbf{0}_{N \times 1},\textbf{I}_{N} \right) $ is the additive noise at the $\lambda$-th relay and $P_{1}$ is the total power used at the source node for every channel use. Using the $ N = \lceil \frac{2^{a-1}-a}{2} \rceil + 2a$ length vector, $\textbf{r}_{\lambda}$, the $\lambda$-th relay constructs the $2^{a-1} + a$ length new vector $\bar{\textbf{r}}_{\lambda}$ given by \eqref{concatenation} shown at the top of this page, 
where $\textbf{r}_{\lambda}(i)$ denotes the $i$-th component of the vector $\textbf{r}_{\lambda}$. The $\lambda$-th relay is assumed to obtain a perfect estimate of the phase component of $h_{\lambda}$ using the training symbols sent during the first $\lceil \frac{2^{a-1}-a}{2} \rceil$ channel uses in the first phase. This has enabled the phase compensation in \eqref{concatenation} which can also be given by 
\begin{equation*}
\bar{\textbf{r}}_{\lambda} = \sqrt{P_{1}N}|h_{\lambda}|\bar{\textbf{x}} + \bar{\textbf{n}}_{\lambda}
\end{equation*}
where $\bar{\textbf{x}} = [ \underbrace{\alpha ~\alpha ~\cdots ~\alpha}_{2^a-a ~\mbox{ times }} ~\hat{x}_{1}~\hat{x}_{2} \cdots \hat{x}_{2a}]^{T} \in \mathbb{C}^{(2^{a-1} + a) \times 1}.$ 
Note that the concatenating operation in \eqref{concatenation} continues to keep the components of $\bar{\textbf{n}}_{\lambda}$ identically distributed and uncorrelated to each other.

In the second phase, all the relay nodes are scheduled to transmit $T$ length vectors to the destination simultaneously. In \cite{HRH}, we used arbitrary values for $T,$ and in this paper since we make the destination see a square design, henceforth, we take $T=2^a.$

Each relay is equipped with a fixed pair of matrices $\textbf{A}_{\lambda}$, $\textbf{B}_{\lambda} \in \mathbb{C}^{2^a \times (2^{a-1} + a)}$ and is allowed to linearly process the vector $\bar{\textbf{r}}_{\lambda}$. The $\lambda$-th relay is scheduled to transmit
\begin{equation}
\label{amlify_technique}
\textbf{t}_{\lambda} = \sqrt{\frac{P_{2}2^a}{P_{r}}}\left\lbrace \textbf{A}_{\lambda}\bar{\textbf{r}}_{\lambda} + \textbf{B}_{\lambda}\bar{\textbf{r}}_{\lambda}^{*}\right\rbrace \in \mathbb{C}^{2^a \times 1},
\end{equation}
where $P_{2}$ is the total power used at each relay for every channel use in the second phase and $P_{r}$ is the average norm of the vector $\bar{\textbf{r}}_{\lambda}$. Note that $P = P_{1} + KP_{2}$ becomes the total power transmitted by all the nodes for a single channel use. The vector received at the destination is given by
\begin{equation*}
\textbf{y} = \sum_{\lambda = 1}^{2^a} g_{\lambda}\textbf{t}_{\lambda} + \textbf{w} \in \mathbb{C}^{2^a \times 1},
\end{equation*}
\noindent where $\textbf{w} \sim \mathcal{CSCG} \left(\boldsymbol{0}_{2^a \times 1},\textbf{I}_{2^a} \right)$ is the additive noise at the destination. Substituting for $\textbf{t}_{\lambda}$, $\textbf{y}$ can be written as
\begin{equation*}
\textbf{y} = \sqrt{\frac{P_{1}P_{2} N2^a}{P_{r}}}\textbf{X}\textbf{g}  + \textbf{n},
\end{equation*}
\noindent where
\begin{itemize}
\item $\textbf{n} = \sqrt{\frac{P_{2}2^a}{P_{r}}}\left[ \sum_{\lambda=1}^{2^a} g_{\lambda}\left\lbrace \textbf{A}_{\lambda} \bar{\textbf{n}}_{\lambda} + \textbf{B}_{\lambda} \bar{\textbf{n}}_{\lambda}^{*}\right\rbrace \right]  + \textbf{w} \in \mathbb{C}^{2^a \times 1}   \label{bfN}.$
\item The equivalent channel \textbf{g} is given by $[|h_{1}|g_{1} ~ |h_{2}|g_{2} ~ \cdots ~ |h_{2^a}|g_{2^a} ]^{T} \in \mathbb{C}^{2^a \times 1}.$
\item Every codeword $\textbf{X} \in \mathbb{C}^{2^a \times 2^a}$ which is of the form \eqref{code_word} (shown at the top of this page) is a function of the information vector $\textbf{x}$ through $\bar{\textbf{x}}$.
\end{itemize}

The covariance matrix $\textbf{R} \in \mathbb{C}^{2^a \times 2^a}$ of the noise vector $\textbf{n}$ is given in \eqref{covariance_matrix} (top of the next page). Note that $\textbf{R}$ depends on the choice of the relay matrices $\textbf{A}_\lambda$ and $\textbf{B}_\lambda.$ The relay matrices need to be chosen such that the resulting code seen by the destination is SSD. 

The Maximum Likelihood (ML) decoder decodes for a vector $\hat{\textbf{x}}$ given in \eqref{ML} (shown at the top of the next page).
\section{Construction of TE-CIODs}
\label{sec3}
In this section, we present the construction of the class of TE-CIODs.
\begin{definition}
\label{def_alpha_cod}
Let the $2^{a-1} \times 2^{a-1}$ matrix $\textbf{X}_1$ represent a COD in $a$ complex variables \cite{TJC}. If the zeros in the design $\textbf{X}_1$ are replaced by a non-zero constant say $\alpha \in \mathbb{C}$, then we refer $\textbf{X}_1$ as a square TE-COD \cite{HRH}.
\end{definition}

The above definition holds both for  the classes of  square CODs as well as non-square CODs which are discussed in detail in \cite{HRH}. In this paper, only square TE-CODs are relevant.

\begin{example}
\label{example_1}
For the well known $4 \times 4$ COD \cite{TJC},  in  the variables  $x_{1}$, $x_{2}$ and $x_{3},$ the corresponding TE-COD is given by,
\begin{equation}
\textbf{X}_{\mbox{TE-COD}} = \left[\begin{array}{rrrr}
x_{3} & \alpha & x_{2} & x_{1}\\
\alpha & x_{3} & x_{1}^{*} & -x_{2}^{*}\\
x_{2}^{*} & x_{1} & -x_{3}^{*} & \alpha\\
x_{1}^{*} & -x_{2} & \alpha & -x_{3}^{*}\\
\end{array}\right].\\
\end{equation}
\end{example}

Given a $2^{a-1} \times 2^{a-1}$ TE-COD, $\textbf{X}_1,$ in $a$ variables, note that every column of $\textbf{X}_1$ contains exactly $a$ distinct variables and $2^{a-1} - a$ copies of $\alpha$.

For a given set of $2a$ complex variables $\{x_1,x_2, \cdots, x_{2a}\},$ let 
the complex variables $\tilde{x}_{1}, \tilde{x}_{2} \cdots \tilde{x}_{2a}$ be defined  as
\begin{equation}
\label{interleaving}
\tilde{x}_{m} = x_{mI} +{\bf i} x_{(m+a)Q} \mbox{ and }\tilde{x}_{m+a} = x_{(m+a)I} +{\bf i} x_{mQ}
\end{equation}
for all $m= 1$ to $a,$ where $x_m=x_{iI}+{\bf i}x_{iQ}.$ Notice that the new variables are nothing but the so called {\it coordinate interleaved variables} extensively used in \cite{KhR}.

If $\textbf{X}_{1}$ and $\textbf{X}_{2}$ represents two identical square TE-CODs for $2^{a-1}$ antennas (relays) in variables $\tilde{x}_{1}, \tilde{x}_{2} \cdots \tilde{x}_{a}$ and $\tilde{x}_{a+1}, \tilde{x}_{a+2} \cdots \tilde{x}_{2a}$ respectively, then a square TE-CIOD for $2^{a}$ antennas (relays) is constructed as 
\begin{equation*}
\textbf{X}_{\mbox{TE-CIOD}} = \left[\begin{array}{cccccccc}
\textbf{X}_{1} & \textbf{0}_{2^{a-1} \times 2^{a-1}}\\ 
\textbf{0}_{2^{a-1} \times 2^{a-1}} & \textbf{X}_{2}\\ 
\end{array}\right].
\end{equation*}
\indent Note that every column of $\textbf{X}_{\mbox{TE-CIOD}}$ has $a$ distinct complex variables,  $2^{a-1} - a$ copies of $\alpha,$ and $2^{a-1}$ zeros. Therefore, with a $(2^{a-1}+a)$-length vector $\bar{\textbf{x}}$ given by
\begin{equation*}
\bar{\textbf{x}} = [\underbrace{\alpha ~\alpha ~\cdots ~\alpha}_{2^{a-1} - a ~\mbox{ times }} ~\tilde{x}_{1} ~\tilde{x}_{2} ~\cdots ~\tilde{x}_{2a}]^{T} \in \mathbb{C}^{(2^{a-1}+a) \times 1},
\end{equation*}
the design $\textbf{X}_{\mbox{TE-CIOD}}$ can be written in its column vector representation \cite{Lia} as

{\small
\begin{equation}
\label{alpha_ciod}
\textbf{X}_{\mbox{TE-CIOD}}  = \left[ \textbf{C}_{1}\bar{\textbf{x}} + \textbf{D}_{1}\bar{\textbf{x}}^{*} ~~ \textbf{C}_{2}\bar{\textbf{x}} + \textbf{D}_{2}\bar{\textbf{x}}^{*} ~~ \cdots ~~ \textbf{C}_{2^a}\bar{\textbf{x}} + \textbf{D}_{2^a}\bar{\textbf{x}}^{*} \right],
\end{equation}
}

\noindent where $\textbf{C}_{\lambda}, \textbf{D}_{\lambda} \in \mathbb{C}^{2^{a} \times (2^{a-1} + a)},$ $k=1,2,\cdots,2^a,$ are the column-vector representation matrices of $\textbf{X}_{\mbox{TE-CIOD}}$. 

The following theorem provides two important relations satisfied by the matrices $\textbf{C}_{\lambda}, \textbf{D}_{\lambda}$ of TE-CIODs.
\begin{theorem}
\label{thm2}
The column-vector representation matrices $\textbf{C}_{\lambda}, \textbf{D}_{\lambda}$ of a TE-CIOD, $\textbf{X}_{\mbox{TE-CIOD}},$ can be chosen to satisfy the following relations,
{\small
\begin{equation}
\label{condition1}
\textbf{C}_{\lambda}\textbf{C}^{H}_{\lambda} + \textbf{D}_{\lambda}\textbf{D}^{H}_{\lambda} = \left[\begin{array}{cc}
\textbf{I}_{2^{a-1}} & \textbf{0}_{2^{a-1} \times 2^{a-1}}\\
\textbf{0}_{2^{a-1} \times 2^{a-1}}  & \textbf{0}_{2^{a-1} \times 2^{a-1}}\\
\end{array}\right]~\forall ~\lambda  = 1 \mbox{ to }2^{a-1},
\end{equation}
}

and
{\small
\begin{equation}
\label{condition2}
\textbf{C}_{\lambda}\textbf{C}^{H}_{\lambda} + \textbf{D}_{\lambda}\textbf{D}^{H}_{\lambda} = \left[\begin{array}{cc}
\textbf{0} & \textbf{0}_{2^{a-1} \times 2^{a-1}}\\
\textbf{0}_{2^{a-1} \times 2^{a-1}}  & \textbf{I}_{2^{a-1}}\\
\end{array}\right]
\end{equation}
 ~~~~~~~~~~~~~~~~~~~~~~~~~~~~$\forall ~\lambda  = 2^{a-1} + 1 \mbox{ to }2^{a}$.
}
\end{theorem}
\begin{proof}
Since the entries of a $\textbf{X}_{\mbox{TE-CIOD}}$ are of the form $\alpha$, $\pm \tilde{x}_{i}$ and $\pm \tilde{x}^{*}_{i} ~\forall~i = 1$ to $2a$ and the vector $\bar{\textbf{x}}$ is given by $\bar{\textbf{x}} = [ \underbrace{\alpha ~\alpha ~\cdots ~\alpha}_{2^{a-1} - a ~times} ~\tilde{x}_{1}~\tilde{x}_{2} \cdots \tilde{x}_{2a}]^{T}$, it is straightforward to verify that the matrices $\textbf{C}_{\lambda}, \textbf{D}_{\lambda}$ satisfy the following three properties,
\begin{itemize}
\item The entries of the matrices $\textbf{C}_{\lambda}, \textbf{D}_{\lambda}$ are $0, \pm 1$.
\item The matrices $\textbf{C}_{\lambda}, \textbf{D}_{\lambda}$ can have at most one non-zero entry in every row.
\item The matrices $\textbf{C}_{\lambda}$ and $\textbf{D}_{\lambda}$ do not contain non-zero entries in the same row.
\end{itemize}
Note that since TE-CIODs are constructed using TE-CODs (using a block diagonal structure), out of the $2a$ complex variables, only $a$ number of them appear exactly once (either as $\pm \tilde{x}_{i}$ or $\pm \tilde{x}^{*}_{i}$) in every column of the design. In particular, the variables $\tilde{x}_{1}, \tilde{x}_{2} \cdots \tilde{x}_{a}$ appear only in the first $2^{a-1}$ rows and the first $2^{a-1}$columns of the design where as the variables $\tilde{x}_{a+1}, \tilde{x}_{a+2} \cdots \tilde{x}_{2a}$ appear only in the last $2^{a-1}$ columns and the last $2^{a-1}$ rows of the design. We only provide a proof for the relation in \eqref{condition1}. Since a TE-CIOD is block diagonal, the relation in \eqref{condition2} can be proved on the similar lines of that of \eqref{condition1}. Without loss of generality, let us assume that $l$ out of the $a$ complex variables which appear in the $\lambda$-th column (for $1 \leq k \leq 2^{a-1}$) of the design are of the form $\pm \tilde{x}_{i}$. With such an assumption, the matrix $\textbf{C}_{\lambda}$ must have $2^{a-1} - a +l$ non-zero rows (where $l$ non-zero rows are for the variables and the rest are for the $\alpha$'s). Further, as the remaining $a-l$ variables appear as conjugates (i.e., of the form $\pm \tilde{x}^{*}_{i}$), the matrix $\textbf{D}_{\lambda}$ must have $a-l$ non-zero rows. Since there are $2^{a-1} - a$ copies of $\alpha$ in the vector $\bar{\text{x}}$, the non-zero entries in the $2^{a-1} - a$ non-zero rows (which are alloted for the $2^{a-1} - a$ copies of $\alpha$) of $\textbf{C}_{\lambda}$ can be chosen to appear in different columns. Therefore, each of the first $2^{a-1}$ columns of $\textbf{C}_{\lambda}$ and $\textbf{D}_{\lambda}$ will have exactly one non-zero entry. Since the variables $\tilde{x}_{a+1}, \tilde{x}_{a+2} \cdots \tilde{x}_{2a}$ do not appear in the first $2^{a-1}$ columns of the design, each of the last $2^{a-1}$ columns of $\textbf{C}_{\lambda}$ and $\textbf{D}_{\lambda}$ are zeros. Hence the relay matrices  satisfy the relation in \eqref{condition1}.
\end{proof}

\indent The formal definition of TE-CIOD is as follows:
\begin{definition} \label{def_pdstbc} \noindent The collection $\mathcal{C}$ of $2^a \times 2^a$ codeword matrices given by \eqref{code_word},
\begin{equation}
\label{dstbc}
\mathcal{C} = \left\{ \textbf{X} \mid \forall ~\textbf{x} \in \mathcal{S} \right\}
\end{equation}

\noindent is called a Training-Embedded Coordinate Interleaved Orthogonal Design  (TE-CIOD) which is determined by the sets $\left\lbrace \textbf{P}, \textbf{Q}, \textbf{A}_{\lambda}, \textbf{B}_{\lambda}\right\rbrace$ and $\mathcal{S},$ where the column vector representation matrices $\textbf{C}_{\lambda}$ and $\textbf{D}_{\lambda}$ of $\textbf{X}_{\mbox{TE-CIOD}}$ given in \eqref{alpha_ciod} are  used as the relay matrices, $\textbf{A}_{\lambda}$ and  $\textbf{B}_{\lambda}$ respectively.
\end{definition}
\indent Note that unlike the existing DSTBCs, TE-DSTBCs contain the training symbols in the code structure along with the information symbols justifying their name. Also, note that unlike TE-CODs, not all the entries of TE-CIODs are non-zero. In particular, TE-CIODs, have block diagonal structure. Due to the block diagonal structure, the number of $\alpha$'s transmitted from the source to the relays is much lesser compared to that of TE-CODs for the same number of relays. In the following section,  we show that this training-embedding enables SSD property TE-CIODs.\\
\indent In the following example, we explicitly a TE-CIOD for a network with 8 relays.
\begin{example}
\label{example_4}
The $8 \times 8$ TE-CIOD is given by,
\begin{equation}
\label{CIOD_8_antennas}
 \left[\begin{array}{rrrrrrrr}
\tilde{x}_{3} & \alpha & \tilde{x}_{2} & \tilde{x}_{1} & 0 & 0 & 0 & 0\\
\alpha & \tilde{x}_{3} & \tilde{x}_{1}^{*} & -\tilde{x}_{2}^{*}  & 0 & 0 & 0 & 0\\
\tilde{x}_{2}^{*} & \tilde{x}_{1} & -\tilde{x}_{3}^{*} & \alpha  & 0 & 0 & 0 & 0\\
\tilde{x}_{1}^{*} & -\tilde{x}_{2} & \alpha & -\tilde{x}_{3}^{*}  & 0 & 0 & 0 & 0\\
0 & 0 & 0 & 0 & \tilde{x}_{6} & \alpha & \tilde{x}_{5} & \tilde{x}_{4}\\
0 & 0 & 0 & 0 & \alpha & \tilde{x}_{6} & \tilde{x}_{4}^{*} & -\tilde{x}_{5}^{*}\\
0 & 0 & 0 & 0 & \tilde{x}_{5}^{*} & \tilde{x}_{4} & -\tilde{x}_{6}^{*} & \alpha\\
0 & 0 & 0 & 0 & \tilde{x}_{4}^{*} & -\tilde{x}_{5} & \alpha & -\tilde{x}_{6}^{*}\\
\end{array}\right]
\end{equation}
where
$\tilde{x}_{m} = x_{mI} +{\bf i} x_{(m+3)Q} \mbox{ and }\tilde{x}_{m+3} = x_{(m+3)I} +{\bf i} x_{mQ}$ for $m= 1$ to $3$.
The set of information vectors equipped at the source is given by
\begin{equation*}
\mathcal{S} = \{ \textbf{x} = [\alpha ~x_{1} ~x_{2} ~x_{3} ~x_{4} ~x_{5} ~x_{6}]^T ~|~ \forall ~ x_{i} \in \mathcal{M} \}.
\end{equation*}
The precoding matrices are given by
\begin{equation*}
\textbf{P} = \frac{1}{2}\left[\begin{array}{cccccccc}
1 & 0 & 0 & 0 & 0 & 0 & 0 \\
0 & 1 & 0 & 0 & 1 & 0 & 0 \\
0 & 0 & 1 & 0 & 0 & 1 & 0 \\
0 & 0 & 0 & 1 & 0 & 0 & 1 \\
0 & 1 & 0 & 0 & 1 & 0 & 0 \\
0 & 0 & 1 & 0 & 0 & 1 & 0 \\
0 & 0 & 0 & 1 & 0 & 0 & 1 \\
\end{array}\right],
\end{equation*}

\begin{equation*}
~ \textbf{Q} = \frac{1}{2}\left[\begin{array}{rrrrrrrr}
0 & 0 & 0 & 0 & 0 & 0 & 0 \\
0 & 1 & 0 & 0 & -1 & 0 & 0 \\
0 & 0 & 1 & 0 & 0 & -1 & 0 \\
0 & 0 & 0 & 1 & 0 & 0 & -1 \\
0 & -1 & 0 & 0 & 1 & 0 & 0 \\
0 & 0 & -1 & 0 & 0 & 1 & 0 \\
0 & 0 & 0 & -1 & 0 & 0 & 1 \\
\end{array}\right].
\end{equation*}
The relay matrices $\textbf{A}_{\lambda}, \textbf{B}_{\lambda} \in \mathcal{C}^{8 \times 7}$ required to construct the $8 \times 8$ TE-CIOD are given by
\begin{equation*}
\textbf{A}_{i} = \left[\begin{array}{cc}
\textbf{E}_{i}\\
\textbf{0}_{4 \times 7}\\
\end{array}\right], ~\textbf{B}_{i} = \left[\begin{array}{cc}
\textbf{F}_{i}\\
\textbf{0}_{4 \times 7}\\
\end{array}\right] \mbox{ for } 1 \leq i \leq 4 \mbox{ and }
\end{equation*}
\begin{equation*}
\textbf{A}_{i} = \left[\begin{array}{cc}
\textbf{0}_{4 \times 7}\\
\textbf{E}_{i}\\
\end{array}\right], ~\textbf{B}_{i} = \left[\begin{array}{cc}
\textbf{0}_{4 \times 7}\\
\textbf{F}_{i}\\
\end{array}\right] \mbox{ for } 5 \leq i \leq 8,
\end{equation*}
where
\begin{equation*}
\textbf{E}_{1} = \left[\begin{array}{rrrrrrrr}
0 & 0 & 0 & 1 & 0 & 0 & 0\\
1  & 0 & 0 & 0 & 0 & 0 & 0\\
0 & 0 & 0 & 0 & 0 & 0 & 0\\
0 & 0 & 0 & 0 & 0 & 0 & 0\\
\end{array}\right] ;~ 
\end{equation*}
\begin{equation*}
\textbf{F}_{1} = \left[\begin{array}{rrrrrrrr}
0 & 0 & 0 & 0 & 0 & 0 & 0\\
0  & 0 & 0 & 0 & 0 & 0 & 0\\
0 & 0 & 1 & 0 & 0 & 0 & 0\\
0 & 1 & 0 & 0 & 0 & 0 & 0\\
\end{array}\right];
\end{equation*}
\begin{equation*}
\textbf{E}_{2} = \left[\begin{array}{rrrrrrrr}
1 & 0 & 0 & 0  & 0 & 0 & 0\\
0  & 0 & 0 & 1 & 0 & 0 & 0\\
0 & 1 & 0 & 0 & 0 & 0 & 0\\
0 & 0 & -1 & 0 & 0 & 0 & 0\\
\end{array}\right];~ 
\textbf{F}_{2} = \textbf{0}_{4 \times 7};
\end{equation*}
\begin{equation*}
\textbf{E}_{3} = \left[\begin{array}{rrrrrrrr}
0 & 0 & 1 & 0 & 0 & 0 & 0\\
0  & 0 & 0 & 0 & 0 & 0 & 0\\
0 & 0 & 0 & 0 & 0 & 0 & 0\\
1 & 0 & 0 & 0 & 0 & 0 & 0\\
\end{array}\right] ;~ 
\end{equation*}
\begin{equation*}
\textbf{F}_{3} = \left[\begin{array}{rrrrrrrr}
0 & 0 & 0 & 0 & 0 & 0 & 0\\
0  & 1 & 0 & 0 & 0 & 0 & 0\\
0 & 0 & 0 & -1 & 0 & 0 & 0\\
0 & 0 & 0 & 0 & 0 & 0 & 0\\
\end{array}\right];
\end{equation*}
\begin{equation*}
\textbf{E}_{4} = \left[\begin{array}{rrrrrrrr}
0 & 1 & 0 & 0 & 0 & 0 & 0\\
0  & 0 & 0 & 0 & 0 & 0 & 0\\
1 & 0 & 0 & 0 & 0 & 0 & 0\\
0 & 0 & 0 & 0 & 0 & 0 & 0\\
\end{array}\right];~
\end{equation*}
\begin{equation*}
\textbf{F}_{4} = \left[\begin{array}{rrrrrrrr}
0 & 0 & 0 & 0 & 0 & 0 & 0\\
0  & 0 & -1 & 0 & 0 & 0 & 0\\
0 & 0 & 0 & 0 & 0 & 0 & 0\\
0 & 0 & 0 & -1 & 0 & 0 & 0\\
\end{array}\right].
\end{equation*}
\begin{equation*}
\textbf{E}_{5} = \left[\begin{array}{rrrrrrrr}
0 & 0 & 0 & 0 & 0 & 0 & 1\\
1  & 0 & 0 & 0 & 0 & 0 & 0\\
0 & 0 & 0 & 0 & 0 & 0 & 0\\
0 & 0 & 0 & 0 & 0 & 0 & 0\\
\end{array}\right] ;~ 
\end{equation*}
\begin{equation*}
\textbf{F}_{5} = \left[\begin{array}{rrrrrrrr}
0 & 0 & 0 & 0 & 0 & 0 & 0\\
0  & 0 & 0 & 0 & 0 & 0 & 0\\
0 & 0 & 0 & 0 & 0 & 1 & 0\\
0 & 0 & 0 & 0 & 1 & 0 & 0\\
\end{array}\right];
\end{equation*}
\begin{equation*}
\textbf{E}_{6} = \left[\begin{array}{rrrrrrrr}
1 & 0 & 0 & 0 & 0 & 0 & 0\\
0  & 0 & 0 & 0 & 0 & 0 & 1\\
0 & 0 & 0 & 0 & 1 & 0 & 0\\
0 & 0 & 0 & 0 & 0 & -1 & 0\\
\end{array}\right];~ \textbf{F}_{6} = \textbf{0}_{4 \times 7};
\end{equation*}
\begin{equation*}
\textbf{E}_{7} = \left[\begin{array}{rrrrrrrr}
0 & 0 & 0 & 0 & 0 & 1 & 0\\
0  & 0 & 0 & 0 & 0 & 0 & 0\\
0 & 0 & 0 & 0 & 0 & 0 & 0\\
1 & 0 & 0 & 0 & 0 & 0 & 0\\
\end{array}\right] ;~ 
\end{equation*}
\begin{equation*}
\textbf{F}_{7} = \left[\begin{array}{rrrrrrrr}
0 & 0 & 0 & 0 & 0 & 0 & 0\\
0  & 0 & 0 & 0 & 1 & 0 & 0\\
0 & 0 & 0 & 0 & 0 & 0 & -1\\
0 & 0 & 0 & 0 & 0 & 0 & 0\\
\end{array}\right];
\end{equation*}
\begin{equation*}
\textbf{E}_{8} = \left[\begin{array}{rrrrrrrr}
0 & 0 & 0 & 0 & 1 & 0 & 0\\
0  & 0 & 0 & 0 & 0 & 0 & 0\\
1 & 0 & 0 & 0 & 0 & 0 & 0\\
0 & 0 & 0 & 0 & 0 & 0 & 0\\
\end{array}\right] \mbox{ and }~
\end{equation*}
\begin{equation*}
\textbf{F}_{8} = \left[\begin{array}{rrrrrrrr}
0 & 0 & 0 & 0 & 0 & 0 & 0\\
0  & 0 & 0 & 0 & 0 & -1 & 0\\
0 & 0 & 0 & 0 & 0 & 0 & 0\\
0 & 0 & 0 & 0 & 0 & 0 & -1\\
\end{array}\right].
\end{equation*}
The number of channel uses in the first phase and second phase are $7$ and $8,$ respectively. Therefore, the rate of the scheme is $\frac{6}{15}$ complex symbols per channel use.
\end{example}
\begin{table*}
\caption{Rate of the proposed codes are listed for certain number of relays. The symbol $L$ denotes the total number of channel uses including the two phases to construct the designs}
\begin{center}
\begin{tabular}{|c|c|c|c|c|c|c|c|c|c|c|}
\hline  & K = 2 & K = 4 & K = 8 & K = 16 & K = 32\\
\hline $R_{\mbox{square ~TE-CODs}}$  & $\frac{1}{2}$ & $\frac{3}{8}$ & $\frac{4}{14}$ & $\frac{5}{27}$ & $\frac{6}{51}$\\
\hline $L_{\mbox{square ~TE-CODs}}$ & 4 & 8 & 14 & 27 & 51\\
\hline $R_{\mbox{TE-CIODs}}$  & $\frac{1}{2}$ & $\frac{1}{2}$ & $\frac{6}{15}$ & $\frac{8}{26}$ & $\frac{10}{48}$\\

\hline $L_{\mbox{TE-CIODs}}$ & 4 & 8 & 15 & 26 & 48\\ \hline 
\end{tabular} 
\end{center}
\hrule
\label{rate_table_designs}
\end{table*}
\section{SSD, Full-diversity and Rate of TE-CIODs}
\label{sec4}
In this section, we discuss the SSD and full-diversity properties of TE-CIODs and also present comparison of the rate with TE-CODs.
\subsection{SSD Property of TE-CIODs}
\label{sec4_subsec1}
From the results of Theorem \ref{thm2}, the covariance matrix $\textbf{R}$ given in \eqref{covariance_matrix} will not be a scaled identity matrix but a diagonal matrix such that $[\textbf{R}]_{i,i} = [\textbf{R}]_{j,j}$ for $1 \leq i, j \leq 2^{a-1}$ and $[\textbf{R}]_{i,i} = [\textbf{R}]_{j,j}$ for $2^{a-1}+1 \leq i, j \leq 2^{a}$. It can be verified that such a structure on $\textbf{R}$ along with the block diagonal structure of the design ensures that every complex-symbol can be ML decoded independent of others. We illustrate this below, for the TE-CIOD code of Example \ref{example_4} for $8$ relays. The $\textbf{R}$ matrix is not a scaled identity matrix, however, the matrix $\textbf{X}^{H}\textbf{R}^{-1}\textbf{X}$ can be written as 
\begin{equation*}
\textbf{X}^{H}\textbf{R}^{-1}\textbf{X} = \left[\begin{array}{cc}
\textbf{X}_{1}^{H}\textbf{R}_{1}^{-1}\textbf{X}_{1} & \textbf{0}_{2^{a-1} \times 2^{a-1}}\\
\textbf{0}_{2^{a-1} \times 2^{a-1}} & \textbf{X}_{2}^{H}\textbf{R}_{2}^{-1}\textbf{X}_{2}\\
\end{array}\right],
\end{equation*}
where the matrices $\textbf{R}_{1}$ and $\textbf{R}_{2}$ are scaled identity matrices. Therefore, $\textbf{X}^{H}\textbf{R}^{-1}\textbf{X}$ becomes
\begin{equation*}
\textbf{X}^{H}\textbf{R}^{-1}\textbf{X} = \left[\begin{array}{cc}
\textbf{R}_{1}^{-1}\textbf{X}_{1}^{H}\textbf{X}_{1} & \textbf{0}_{2^{a-1} \times 2^{a-1}}\\
\textbf{0}_{2^{a-1} \times 2^{a-1}} & \textbf{R}_{2}^{-1}\textbf{X}_{2}^{H}\textbf{X}_{2}\\
\end{array}\right],
\end{equation*}
where the matrices $\textbf{X}_{1}^{H}\textbf{X}_{1}$ and $\textbf{X}_{2}^{H}\textbf{X}_{2}$ are given by \eqref{XHX_CIOD_X1} and \eqref{XHX_CIOD_X2} respectively (shown at the top of the next page). 

Notice that with $\alpha=0,$ these matrices reduce to the one corresponding to the original CIODs reported in \cite{KhR}.
\begin{figure*}
\begin{equation}
\label{XHX_CIOD_X1}
\textbf{X}_{1}^{H}\textbf{X}_{1} = \left[\begin{array}{cccccccc}
|\alpha|^{2} + \sum_{i = 1}^{3} |\tilde{x}_{i}|^{2} & 2\mbox{Re}(\tilde{x}_{3}^{*}\alpha) & 2\mbox{Re}(\tilde{x}_{1}^{*}\alpha^{*}) & 2{\bf i}\mbox{Im}(\tilde{x}_{2}\alpha)\\
* & |\alpha|^{2} + \sum_{i = 1}^{3} |\tilde{x}_{i}|^{2} & 2{\bf i}\mbox{Im}(\tilde{x}_{2}\alpha^{*}) & 2\mbox{Re}(\tilde{x}_{1}\alpha^{*})\\
* & * & |\alpha|^{2} + \sum_{i = 1}^{3} |\tilde{x}_{i}|^{2} & -2\mbox{Re}(\tilde{x}_{3}\alpha)\\
* & * & * & |\alpha|^{2} + \sum_{i = 1}^{3} |\tilde{x}_{i}|^{2}\\
\end{array}\right].
\end{equation}
\begin{equation}
\label{XHX_CIOD_X2}
\textbf{X}_{2}^{H}\textbf{X}_{2} = \left[\begin{array}{cccccccc}
|\alpha|^{2} + \sum_{i = 1}^{3} |\tilde{x}_{i}|^{2} & 2\mbox{Re}(\tilde{x}_{6}^{*}\alpha) & 2\mbox{Re}(\tilde{x}_{4}^{*}\alpha^{*}) & 2{\bf i}\mbox{Im}(\tilde{x}_{5}\alpha)\\
* & |\alpha|^{2} + \sum_{i = 1}^{3} |\tilde{x}_{i}|^{2} & 2{\bf i}\mbox{Im}(\tilde{x}_{5}\alpha^{*}) & 2\mbox{Re}(\tilde{x}_{4}\alpha^{*})\\
* & * & |\alpha|^{2} + \sum_{i = 1}^{3} |\tilde{x}_{i}|^{2} & -2\mbox{Re}(\tilde{x}_{6}\alpha)\\
* & * & * & |\alpha|^{2} + \sum_{i = 1}^{3} |\tilde{x}_{i}|^{2}\\
\end{array}\right].
\end{equation}
\hrule
\end{figure*}
\subsection{Full-Diversity of TE-CIODs}
\label{sec4_subsec1_1}
The TE-CIODs provide fully diversity for a special class of two-dimensional signal sets for which the corresponding CIODs \cite{KhR} are fully diverse. This is because the difference of any two codewords of a TE-CIOD is also a difference of two codewords (those with $\alpha=0$)  of the corresponding CIOD.  The conditions on the signal sets for which CIODs provide full diversity has been proved in \cite{KhR} along with illustrative examples. It turns out that the signal sets for which full-diversity is achieved are precisely those in which no two signal points are on a line parallel to the $x-$axis or parallel to the $y-$ axis. We refer the readers to \cite{KhR} for more details and proofs on the choice of such signal sets for which TE-CIODs offer full-diversity. 

\subsection{Rate of TE-CIODs}
\label{sec4_subsec2}
To distributively construct an TE-CIOD for a network with $2^{a}$ relays, the number of channel uses in the first phase and the second phase are $2a + \lceil \frac{2^{a-1} - a}{2} \rceil$ and $2^{a},$ respectively. Hence the rate of the scheme is 
\begin{equation*}
R_{\mbox{TE-CIODs}} = \frac{2a}{2a + \lceil \frac{2^{a-1} - a}{2} \rceil + 2^{a}}.
\end{equation*}

The rate for square TE-CODs with $2^a$ number of relays is 
$ \frac{a + 1}{a + 1 + \lceil \frac{2^{a} - a - 1}{2} \rceil + 2^{a}}$
complex symbols per channel use \cite{HRH}. Comparing with this we see that, for the same number of relays, the rate of TE-CIODs is larger than that of the TE-CODs. This can be easily seen from the Table \ref{rate_table_designs}.
\section{Simulations Results}
\label{sec5}
In this section, we provide the performance comparison (in terms of the symbol error rate (SER) versus $P$, the total power used by all the nodes per channel use) between the DSTBC from TE-CODs and the DSTBC from TE-CIOD for $K = 4$.  For $K = 4$, the rates (in complex symbols per channel use in the second phase) of the DSTBCs from TE-COD and TE-CIOD are respectively $\frac{3}{4}$ and $1.$ Hence, for a fair comparison, we make the bits per channel use (bpcu) in the second phase equal for both codes, in particular, we make it equal to $3$ bpcu for the simulation purpose.
To achieve the common rate of $3$ bpcu in the second phase, the TE-COD and the TE-CIOD respectively employs the 16-QAM signal set $\{-1+{\bf i}, 1+{\bf i},-1-{\bf i}, 1-{\bf i}, -3+{\bf i}, 3+{\bf i},-3-{\bf i}, 3-{\bf i}, -1+3{\bf i}, 1+3{\bf i},-1-3{\bf i}, 1-3{\bf i}, -3+3{\bf i}, 3+3{\bf i},-3-3{\bf i}, 3-3{\bf i}\}$ and a rotated version of the $8$-QAM signal set $\{-3+{\bf i}, -1+{\bf i}, 1+{\bf i},   3+{\bf i}, -3-{\bf i}, -1-{\bf i}, 1-{\bf i}, 3+{\bf i}\}$ to construct the DSTBCs. 
The SER performance of both codes are plotted in Fig. \ref{simulations} which shows that TE-CIOD performs better than TE-COD for $K = 4$. \\

\begin{figure}
\centering
\includegraphics[width=3.5in]{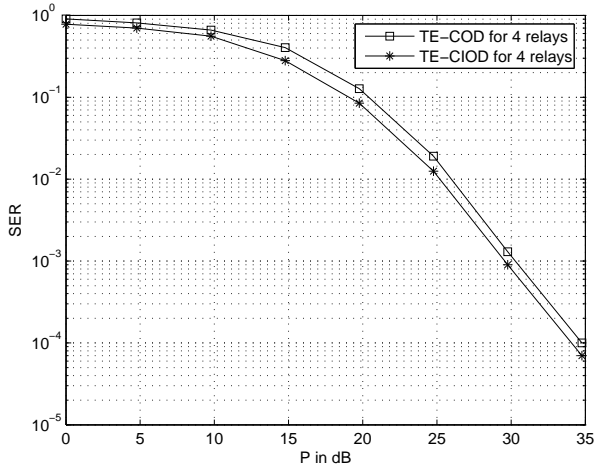}
\caption{SER (symbol error rate) versus P comparison between the DSTBC from TE-COD and the DSTBC from TE-CIOD for $K = 4$ with $3$ bpcu.}
\label{simulations}
\end{figure}
\section{Discussion and Conclusions}
\label{sec6}
In this paper, through a training based distributed space-time coding technique, we have shown to construct the variants of the well known class of CIODs in two-hop relay networks using amplify and forward protocol. This idea can be extended to construct all the multi-group decodable codes \cite{KaR2} existing for point to point co-located MIMO channels in two-hop wireless networks.
\section*{Acknowledgment}
This work was partly supported by the DRDO-IISc Program on Advanced Research in Mathematical Engineering, IISc Bangalore, India, and partly  by the Research Council of Norway through the project 176773/S10 entitled "Optimized Heterogeneous Multiuser MIMO Networks - OptiMO" and the project 183311/S10 entitled "Mobile-to-Mobile Communication Systems (M2M)".


\begin{thebibliography}{1}
                                                                                 
\bibitem{LaW}
J. M. Laneman and G. W. Wornell,  "Distributed space time coded protocols for exploiting cooperative diversity in wireless network'' \emph{IEEE Trans. Inform. Theory}, vol. 49, pp. 2415-2425, Oct. 2003.                                        

\bibitem{JiH1}
Yindi Jing and Babak Hassibi, "Distributed space time coding in wireless relay networks"
\emph{IEEE Trans Wireless communication,} vol. 5, No 12, pp. 3524-3536, Dec. 2006.                                                                                           

\bibitem{YiK}
Zhihang Yi and Il-Min Kim,  "Single-Symbol ML decodable Distributed STBCs for Cooperative Networks,'' \emph{IEEE Trans. Inform. Theory}, vol 53, No 8, pp. 2977 to 2985, August 2007.

\bibitem{HaR}
J. Harshan and B. Sundar Rajan, "High-Rate, Single-Symbol ML Decodable Precoded DSTBCs for Cooperative Networks",  \emph{IEEE Trans. Inform. Theory}, Vol. 55, No. 05, pp. 2004-2015,  May 2009.

\bibitem{YiK2}
Zhihang Yi and Il-Min Kim, "The impact of Noise Correlation and Channel Phase Information on the Data-Rate of the Single-Symbol ML Decodable Distributed STBCs," Submitted to \emph{IEEE Trans. Information theory}.  Available online at arXiv cs.IT/07083387, Aug 2007.

\bibitem{SCR}
D. Sreedhar, A. Chockalingam and B. Sundar Rajan,  "Single-Symbol ML decodable Distributed STBCs for Partially-Coherent Cooperative Networks,'' \emph{IEEE Trans. Wireless Communications}, Vol. 8, No. 5, pp. 2672-2681, May 2009. 

\bibitem{TJC}
V. Tarokh, H. Jafarkhani and A. R. Calderbank, ''Space-time block codes from orthogonal designs", \emph{IEEE Trans. Inform. Theory}, vol. 45, pp.1456-1467, July 1999.

\bibitem{Lia}
Xue-Bin Lia, "Orthogonal Designs with Maximal rates,'' \emph{IEEE Trans. Inform. Theory}, vol. 49, No.10, pp.2468 - 2503, Oct. 2003.

\bibitem{KhR}
Zafar Ali Khan, Md., and B. Sundar Rajan, ''Single Symbol Maximum Likelihood Decodable Linear STBCs'', \emph{IEEE Trans. Inform. Theory}, vol. 52, No. 5, pp.2062-2091, May 2006.

\bibitem{KaR}
Sanjay Karmakar and B. Sundar Rajan, ''Minimum-decoding-complexity maximum-rate space-time block codes from Clifford algebras,'' in the proceedings of \emph{IEEE ISIT}, Seattle, USA,  pp.788-792, July 09-14, 2006,

\bibitem{HRH} 
J. Harshan, B. Sundar Rajan, and Are Hj{\o}rungnes, "High-Rate, Distributed Training-Embedded Complex Orthogonal Designs for Relay Networks,'' submitted to \emph{IEEE Information theory workshop 2010} to be held at Cairo, Egypt. Also available online at arXiv:0908.0051v1[cs.IT] 1st Aug 2009.

\bibitem{KaR2}
Sanjay Karmakar and B. Sundar Rajan, "Multigroup-Decodable STBCs from Clifford Algebras,'' \emph{IEEE Trans. Inform. Theory}, Vol. 55, No. 01, pp. 223-231, Jan. 2009.
\end{thebibliography}
\end{document}